\pdfoutput=1
\documentclass{sig-alternate-05-2015}

\usepackage{balance}  
\usepackage{graphics} 
\usepackage{txfonts}
\usepackage{times}    
\usepackage{hyperref}
\usepackage{color}
\usepackage{textcomp}
\usepackage{booktabs}
\usepackage{ccicons}
\usepackage{todonotes}

\makeatletter
\def\url@leostyle{%
  \@ifundefined{selectfont}{\def\UrlFont{\sf}}{\def\UrlFont{\small\bf\ttfamily}}}
\makeatother
\def\pprw{8.5in}
\def\pprh{11in}

\definecolor{linkColor}{RGB}{6,125,233}
\hypersetup{%
  pdftitle={SIGCHI Conference Proceedings Format},
  pdfauthor={LaTeX},
  pdfkeywords={rank, influence, social network, centrality, reverse engineering},
  bookmarksnumbered,
  pdfstartview={FitH},
  colorlinks,
  citecolor=black,
  filecolor=black,
  linkcolor=black,
  urlcolor=linkColor,
  breaklinks=true,
}

\usepackage{etoolbox}
\makeatletter
\def\@copyrightspace{\relax}
\patchcmd{\maketitle}{\@copyrightspace}{}{}{}
\makeatother

\usepackage{acmcopyright}
\CopyrightYear{2017} 
\setcopyright{acmcopyright}
\conferenceinfo{CSCW '17,}{February 25-March 01, 2017, Portland, OR, USA}
\isbn{978-1-4503-4335-0/17/03}\acmPrice{\$15.00}
\doi{http://dx.doi.org/10.1145/2998181.2998257}

\makeatletter
\newcommand{\removelatexerror}{\let\@latex@error\@gobble}
\makeatother

\usepackage{tikz}
\usetikzlibrary{shapes,snakes}
\usetikzlibrary{calc,trees,positioning,arrows,fit,shapes,calc}
\usepackage{pstricks,pst-node,pst-tree}

\usepackage{url}

\newcommand{\ie}{\emph{i.e.}\xspace}
\newcommand{\eg}{\emph{e.g.}\xspace}

\usepackage[latin1]{inputenc}
\usepackage[francais]{babel}

\usepackage{amssymb}
\usepackage{amsmath}
\usepackage{graphicx}
\usepackage[vlined,linesnumbered]{algorithm2e}
\usepackage{subcaption}

\usepackage{tikz}

\usepackage{color}

\usepackage{ulem}

\newtheorem{definition}{Definition}
\newtheorem{lemma}{Lemma}

\newtheorem{theorem}{Theorem}

\usepackage{flushend}

\begin{document}

\title{Uncovering Influence Cookbooks: Reverse Engineering the Topological Impact in Peer Ranking Services}

  \numberofauthors{2}
  \author{%
    \alignauthor{Erwan Le Merrer}\\
      \affaddr{Technicolor, France}
    \alignauthor{Gilles Tr\'edan}\\
      \affaddr{LAAS/CNRS, France}
  }

\maketitle

\begin{abstract}

Ensuring the early detection of important social network users is a
challenging task. Some peer ranking services are now well established,
such as PeerIndex, Klout, or Kred. Their function is to rank users
according to their influence. This notion of influence is however
abstract, and the algorithms achieving this ranking are
opaque. Following the rising demand for a more transparent web, we
explore the problem of gaining knowledge by reverse engineering such
peer ranking services, with regards to the social network topology
they get as an input.  Since these services exploit the online
activity of users (and therefore their connectivity in social
networks), 
we provide a precise evaluation of how topological metrics of the social network impact the final user ranking.
Our approach
is the following: we first model the ranking service as a black-box
with which we interact by creating user profiles and by performing
operations on them. Through those profiles, we trigger some slight
topological modifications. By monitoring the impact of these
modifications on the rankings of those profiles, we infer the weight of
each topological metric in the black-box, thus reversing the service
influence cookbook.
\end{abstract}


\section{Introduction}
\label{sec:introduction}
The need for an increased transparency in the functioning of
web-services has recently arised, motivated by various use cases such
as privacy or copyright control. For example, work such
as~\cite{chaintreau} proposes to retrieve which piece of information
in a user-profile triggered a particular advertisement. Goal is thus
to infer the internals of black-box services provided by companies on
the web.  Klout or PeerIndex propose to rank users based on their
behavior on social networks (using their social connectivity and
activity). They nevertheless keep secret the algorithms and parameters
used for this ranking\footnote{Those services may provide a score as
an output. Clearly, reversing a ranking function is harder than
reversing a score, as you can obtain a ranking from scores, while the
opposite is impossible.}. This motivated some users to try reversing
their internals~\cite{revklout}. Sometimes information leaks about
some of the ingredients in those hidden recipes; the CEO of PeerIndex for instance admitted to
leverage Pagerank\footnote{\small blog post
on Quora} (and thus graph topological-metrics),
 to compute user intrinsic
influence in a network.  Such an understanding of which metrics are
involved is also of a particular interest for information sharing and
coordination, as it has been shown that some centrality metrics
correlate with the actual ability of network actors to coordinate
others~\cite{cscw, coord}. This knowledge can then serve to assess if
the centrality metric leveraged by the ranking function makes the
ranking service relevant to dispatch roles for given tasks for
example~\cite{coord}.

Nevertheless, reverse engineering such black-boxes is a challenging
task. Indeed, in this web-service paradigm, the user only has access
to the output of the algorithm, and cannot extract any
side-information. Moreover, in many cases such as in peer ranking
services, the user can only take action on a limited part of the
algorithm input.  Motivated by this challenge for transparency, we ask
the following question: \textbf{can a user infer, from the results
  returned by such peer ranking algorithms, what are the topological
  metrics in use, and to which extent?}

We first introduce the ranking service we consider and model our
actions, before warming-up on a toy example. We then generalize the
example and provide a construction to identify the use of a single
arbitrary centrality among a given set of candidates. Then, we assume
that the ranking can be produced by a linear combination of multiple
centralities, and give a generic reverse engineer approach. We
conclude by illustrating such a generic approach on a concrete
scenario, before giving perspectives.

\section{Model \& Warm Up: Reversing One Centrality}

Let us model a social web-service. Each user is represented by
a vertex $v$, together with a set of (possibly unknown to the user) attributes $a(v)$. To
interact with the web-service, users have access to a finite set of actions
$A$. We consider two types of actions: $i)$ \emph{single} actions that only involve a single
user (\eg, posting a message on a wall) and might change part of the user
profile $a(v)$. And $ii)$, \emph{pair} actions that involve a pair of users (\eg, following,
declaring or deleting a ``friendship'' relation). These actions impact the network of relations
among users, that we capture as a graph $G_{\infty}(V,E)$, with $V$ and $E$ being respectively the set of vertices and edges in that graph.

Among the features of this web-service, a ranking of the users is
available. While the internals of the ranking methodology are unknown, each user
accesses its output, that is her own ranking at any time. Let $f$ be the ranking
black-box function.  More specifically, $f$ takes as input the graph
$G_{\infty}$ along with the attributes of its nodes (that is $\{a(v),\forall
v\in V\}$) and assigns each node a score $f(i,a(i)),i\in G_\infty$ from
  which is derived an observable ranking $r$ 
of all graph nodes such that: $\forall i,j \in V(G_\infty)^2, i>_rj$ iff
$f(i,a(i))>f(j,a(j))$, that is: ``node
$i$ is more important (or ``influent'') than node $j$''.

The objective of this paper is to gain knowledge on $f$, and more
specifically to evaluate the impact of each action in $A$ on users
rankings. For a given user, the two main difficulties are that first,
she witnesses only a limited part of the input of $f$ (typically her own
friends in the social graph). Second, the output of $f$ is sparse, as
it only provides nodes with a total order relation (\eg, user $x$ is
better ranked than her neighbor $y$).
In order to try reversing $f$, we assume the \textit{querying} user 
is able to create a set of profiles $V_a$ in
the social service, and have those profiles issue any single action of
$A$. She is also able to achieve any pair action between two profiles
of $V_a$, therefore updating the subgraph of $G_\infty$ induced by
nodes of $V_a$. Those two operations are conducted through API calls,
as it is \eg, observed in practice in Facebook~\cite{sybil}.

As a warm up, let us assume that $f$ leverages exactly one of the following
classic centralities $C_{base} = \{degree,$ $eccentricity,$
$betweenness,$ $Pagerank,$ $closeness \}$~\cite{centralities}. To
determine which one is in use, one user wants to build a small
\textit{query graph} $G_Q$, attached to $G_{\infty}$ (then $G_{\infty} \gets
G_{\infty} \cup G_Q$), in order
to reverse $f$.  
To start our analysis of $f$ on a clean basis, the user
  creates nodes $\in G_Q$ that are strictly identical up to their
  connectivity (\ie, their attributes in $a(v)$ regarding single
  actions such as tweets or posted comments are empty).

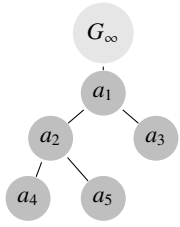
\begin{figure}[t!]
  \centering
\begin{tikzpicture}[shorten >=1pt,auto]
  \tikzstyle{vertex}=[circle,fill=black!25,minimum size=17pt,inner sep=0pt]
  \tikzstyle{graph}=[circle,fill=black!10,minimum size=23pt,inner sep=0pt]
  \node[graph] (g) at (0,0.8) {$G_{\infty}$};
  \node[vertex] (n1) at (0,0) {$a_1$};
  \node[vertex] (n2) at (-0.7,-0.6) {$a_2$};
  \node[vertex] (n3) at (0.7,-0.6) {$a_3$};
  \node[vertex] (n4) at (-1.,-1.4)  {$a_4$};
  \node[vertex] (n5) at (0,-1.4)  {$a_5$};
  \foreach \from/\to in {n1/g,n2/n1,n3/n1,n4/n2,n5/n2}
    \draw (\from) -- (\to);
\end{tikzpicture}
  \caption{A small query graph $G_Q$, solving the single centrality
    reverse engineering problem for plausible set $C_{base}$.}
  \label{fig:sub1}
\end{figure}

\begin{lemma} The query graph $G_Q$ depicted on Figure~\ref{fig:sub1}, of $5$ nodes, is sufficient to reverse engineer a function $f$ that is based on a single centrality $\in C_{base}$, relatively to the other centralities in the same set $C_{base}$.
\end{lemma}
~\\
\begin{proof}
The proof requires showing that such $G_Q$ is able to discriminate the
centralities considered in the set $C_{base}$. Consider graph $G_{\infty}
\cup G_Q$ on Figure~\ref{fig:sub1}. $G_Q$ nodes are given the following ranking, for centralities in 
$C_{base}$\footnote{\small  we conducted numerical
  simulations using the networkx library: \url{https://networkx.github.io/}}:
$<degree, [a_1 =_r a_2 >_r a_3 =_r a_4 =_r a_5]>$, $<eccentricity, [a_1 >_r a_2 =_r a_3 >_r a_4 =_r a_5]>$, $<betweenness, [a_1 >_r a_2 >_r a_3 =_r a_4 =_r a_5]>$, $<Pagerank, [a_2 >_r a_1 >_r a_4 =_r a_5 >_r a_3]>$, $<closeness, [a_1 >_r a_2 >_r a_3 >_r a_4 =_r a_5]>$.  
All rankings are indeed unique, thus allowing to
     designate the centrality used, by user observing 
     rankings produced by $f$ at $G_Q$ nodes she controls.
\end{proof}

Note that 
$G_Q$ is not the unique graph solving this problem instance.

There are obvious interests in minimizing the size of the constructed
query graph: first, constructing a bigger graph requires a longer
time, especially if actions on the service platform are rate-limited
on operations. Second, the bigger the query, the easier it can be
detected by the social service. Note that the graph $G_Q \setminus
a5$, of size $4$ is not a solution, as $degree$ and $betweenness$
produce the same $[a_1 >_r a_2 >_r a_3 =_r a_4 ]$ ranking, as for both
fringe nodes $a_3$ and $a_4$, $betweenness$ is $0$, and $degree$ is
$1$.

\section{General Discrete centrality discrimination}
\label{sec:discr-centr-discr}

We now generalize the reversing logic used on the previous
  example to a set $C$ of arbitrary centralities, possibly in use
  nowadays.  Furthermore, we extend the notion of centrality to the
  one presented in~\cite{borgatti2006graph}: a centrality is
  \textbf{any node-level measure}.

We first draw two observations: discrimination is made by the ranking,
therefore to distinguish between $d$ different centralities one
requires at least $d$ different rankings. Thus $\vert G_Q \vert ! \geq
d$. Second, the discrimination in this set of centralities is made
thanks to graphs we call \textit{delta-reversal graphs}.

\begin{definition}[Delta-reversal graphs]
  $\Delta_{XY}$ is the set of graphs such that $\forall G \in  \Delta_{XY},
  \exists i,j \in V(G) $ s.t. $f_X(G,i)< f_X(G,j) \wedge f_Y(G,i)> f_Y(G,j) $. 
\end{definition}
A delta-reversal graph for two centralities $X$ and $Y$ is a graph
where the ranking $r$ induced by using the ranking provided by $f_X$
(\ie, by a function $f$ only relying on centrality $X$) on the nodes of $G$
would be different than the ranking induced by $f_Y$. Any such graph
would thus allow to discriminate between $X$ and $Y$ being used as
$f$\footnote{Examples of discriminating graphs are known in the
  literature, as they serve as motivation for introducing new
  centralities: see for instance \cite{citeulike:2206413}, where a
  graph is presented that discriminates \textit{random walk betweenness} from
  classic betweenness centrality.}.  The following property is a very handy property for using delta-reversal
graphs.

\begin{definition}[Centrality $k$-locality]
  Let $X$ a centrality. $X$ is said $k$-local if $\forall G_1,G_2$ graphs$,
  \forall i\in V(G_1),j\in V(G_2), V^k(i,G_1)=V^k(j,G_2) \Rightarrow
  f_X(i,G_1)=f_X(j,G_2) $, where $V^k(i,G)$ is the graph induced by the $k$-hop
  neighborhood of $i$ in $G$.
\end{definition}

The intuition is the following: a $k$-local centrality only considers
the $k$-hop neighborhood of a node when assessing it's
importance. This can be seen as the ``scope'' of a centrality: any
topological modification beyond this scope leaves the node importance
unchanged. This can be exploited to join Delta-reversal graphs into
one single query graph while maintaining their discriminating
power. Following this intuition, the following definition states
an important property of those graphs.

\begin{definition}
  Let $G $ a $\Delta_{XY}$ graph, and $dist(i,j)$ the hop-distance
  between nodes $i$ and $j$. If $\exists i,j,k \in G $ s.t. $
  dist(i,k)>\ell \wedge dist(j,k)>\ell \wedge f_X(G,i)< f_X(G,j) \wedge f_Y(G,i)> f_Y(G,j) $ then $G$ is $\ell$-discriminating. $k$ is called
  an anchor.
\end{definition}

\begin{figure}[t!]
 \removelatexerror
  \centering
  \begin{algorithm}[H]
    \KwData{$G_{\infty}$, a target node $a\in V(G_{\infty})$, the set $C$ of suspected centralities ($|C|=d$), $D$ the set of pairwise discriminating graphs for set $C$}
    \KwResult{The centrality $X$ in use in $f$}
    
    $\forall G_{XY} \in D$, let $i_{XY}$, $j_{XY}$ s.t. $f_X(i_{XY}) > f_X(j_{XY}) \land f_Y(i_{XY}) < f_Y(j_{XY})$\;

    \textit{//Building and attaching the general query graph to $G_{\infty}$}\\
    \For{$\forall G \in D$}{ 
      $V(G_{\infty}) \gets V(G_{\infty}) \cup V(G)$\;
      $E(G_{\infty}) \gets E(G_{\infty}) \cup E(G) \cup (a, anchor(G))$\;
    }
    $r \gets r(f(G_{\infty})$\;

    Let $M$ be a $d \times d$ matrix initialized to false\;

    \textit{//Retrieving the centrality in use in $f$}\\
    \For{$a=1 \dots d$}{
          \For{$b=a+1 \dots d$}{
            $M_{a,b} = i_{X_aX_b} >_r j_{X_aX_b}$\;
          }
    }
    Let $s$ be s.t. $\forall k=1 \dots d, M_{s,k}=true$\;

    \Return{$X_s$}\;
    \caption{A reverse engineering algorithm, identifying the centrality in use in arbitrary centrality set $C$.}
    \label{algo-general}
  \end{algorithm}
\end{figure}

\subsection{Combining Delta-reversal graphs}
\label{sec:comb-diff-graphs}

We now explain how to combine pairwise discriminating graphs into a single query graph.
\begin{lemma}
\label{lem:comb-2delta-revers}  Let $X,Y,Z$ three centralities and let $k$ their maximum locality. Then $\forall
G_1 \in \Delta_{X,Y}, G_2 \in \Delta_{X,Z}, G_3\in \Delta_{Y,Z}$, if
all these graphs are $\ell>k$-discriminating,  then $G_S= (V(G_{1} \cup
G_{2}\cup G_{3})\cup\{a\}, E(G_{1} \cup
G_{2} \cup G_3)\cup\{(a,m_1),(a,m_2),(a,m_3)\}) \in \Delta_{XYZ}$.
\end{lemma}

\begin{proof}
  Since $G_1$ is discriminating, $m_1$ exists.  Let $i_1,j_1$  the
  corresponding  anchor nodes. Let $\sigma_X(G_1), \sigma_Y(G_1)$ the ranks
  of centralities $X,Y$. Assume w.l.o.g. that $\sigma_X(G_1,i_1)>\sigma_X(G_1,j_1)$
  and yet $\sigma_Y(G_1,i_1)<\sigma_Y(G_1,j_1)$. Consider $i_1$: we have
  $d(i_1,m_1)>k$ and thus $V^k(i_1,G_1)=V^k(i_1,G_S)$. Thus
  $f_X(i,G_1)=f_X(i,G_S)$. As the same applies for $j_1$ we deduce that
  $\sigma_X(G_S,i_1)>\sigma_X(G_S,j_1)$ and yet
  $\sigma_Y(G_S,i_1)<\sigma_Y(G_S,j_1)$.

  Thus $G_S \in \Delta_{XY}$. A similar reasoning holds for $i_2,j_2$ and
  $i_3,j_3$ thus $G_S \in \Delta_{XZ} \cap \Delta_{YZ}\cap \Delta_{XZ}
  =\Delta_{XYZ} $.
\end{proof}

This lemma is very useful, as it provides us with a way to create
discriminating graphs from pair of known ones. 
The following lemma finally generalizes the construction:

\begin{lemma}
  \label{lem:comb-delta-revers}
  Let $C$ a set of $d$ centralities and let $k$ their maximum locality. Let
  $D=\{G_{AB} \in \Delta_{AB}, \forall A,B\neq A \in C^2\}$, a set containing a pairwise
  discriminating graph for each pair of centrality in $C$.  If
all these graphs are $\ell>k$-differentiated,  then let $G_S=(V(\cup_{G\in D}
G)\cup\{a\}, E(\cup_{G\in D} G)\cup\{(a,m_{AB},\forall A,B\neq A \in C\})$, where
$m_{XY}$ is an anchor of $G_{XY}\in D$. Then
$G_S \in \Delta_{C}$.
\end{lemma}
\begin{proof}(sketch): identical to Lemma~\ref{lem:comb-2delta-revers}.  
\end{proof}
The $G_S$ construction therefore allows for any set of centralities,
given pairwise discriminating graphs, to construct one general
discriminating graph achieving the reverse engineering
process. Note that the complexity is
quadratic: a graph to compare $d$ centralities requires $\Omega(d^2)$ pairwise
discriminating graphs.

We are now ready to propose a general method to infer which
  centrality is in use in $f$. It is shown in
  Algorithm~\ref{algo-general}.

\begin{theorem}
  Let $G_\infty$ a graph, and $r$ an unknown ranking function relying on
  centrality $z$. If $z\in C$ then Algorithm~\ref{algo-general} returns $z$.
\end{theorem}
\begin{proof}
  First, observe that in Algorithm~\ref{algo-general}, lines $1-5$ implement the
  construction of a combined Delta-reversal graph as defined in
  Lemma~\ref{lem:comb-delta-revers}. Line $6$ collects the resulting
  ranking. Consider $M$ at line 12. For $z$ to be correctly identified, two
  conditions must hold:
  $i)$ the line $M_{z,.}$ contains only entries at true, and $ii)$ all other
  lines $M_{i,.},i\neq z$ contain at least one false entry.

  Consider line $M_{z,.}$. By contradiction, assume that one entry, say $j$ is
  false. Then necessarily $i_{z,j}<_rj_{z,j}$ line 11. Since $r$ is obtained
  using $z$, we deduce $f_z(i_{z,j})<f_z(j_{z,j})$. This contradicts the
  definition of $i_{z,j}$ and $j_{z,j}$ that are chosen line 1 in the subgraph
  $G_{zj}$ such that $f_z(i_{z,j})>f_z(j_{z,j})$. We conclude that $M_{z,.}$
  contains only true entries.

  Now, assume there exists another line, say $i$, such that $M_{i,.}$ contains only true
  entries. Consider column $z$: we have $M_{i,z}=$true. As in the previous step,
  we deduce $f_z(i_{i,z})>f_z(j_{i,z})$; this again contradicts the definition
  of $i_{i,z}$ and $j_{i,z}$ chosen line 1 in $G_{iz}$ such that
  $f_z(i_{i,z})<f_z(j_{i,z})$. Thus every other line has at least a negative
  entry.

  Therefore, we conclude that $X_s=z$ line 13: Algorithm~\ref{algo-general} has
  identified $z$.
\end{proof}
The sketch presented in Algorithm~\ref{algo-general} can be optimized in many
ways. First, one can build the query graph incrementally and only test the
relevant centralities: let $G_{ab}$ be the first added Delta-reversal graph line
4 and 5. It is possible to test directly the value of $M_{ab}$. Assume
$M_{a,b}$=False, then necessarily centrality $X_a$ is not used in $f$. There is
therefore no need to add any other $G_{ac},\forall c\in D$ graph.

Second, observe that we focus on pairwise Delta-reversal graphs. Some
Delta-reversal graphs allow to differentiate between more than two centralities
(for instance, the graph $G_Q$ Figure~\ref{fig:sub1} that allows to
differentiate between $5$ centralities at once, while containing only 5
nodes). Using such graphs drastically reduces the size of final the query
graph.

\section{Reverse Engineering a Linear Combination of Centralities}
\label{sec:gener-revers-appr}
In the previous section, we have seen how to identify which centrality
  is used given a finite set of suspected centralities.
We now propose a method for extending to a $f$ that is a
linear combination of suspected centralities, for it allows more complex and subtle ranking functions.

\begin{figure}[t!]
 \removelatexerror
  \centering
  \begin{algorithm}[H]
    \KwData{$G_{\infty}$, a target node $a\in V(G_{\infty})$, operations $\{u_1,\ldots,u_d\}$}
    \KwResult{An estimate of $\mathbf{h}$ (\ie, the vector containing the weight of each centrality in $f$)}
    Let $\mathbf{k}$ be a vector of size $d-1$ initialized to $0$\;
    \For{$1\leq i\leq d$}{ 
      \textit{//attach a query node to target node $a$, and conduct operations over it}\\
      Create node $a_i: V(G_{\infty}) \gets V(G_{\infty})
      \cup\{a_i\}$\; Add edge $(a_i,a): E(G_{\infty}) \gets E(G_{\infty}) \cup\{(a_i,a)\}$\; Apply
      $u_i(a_i)$\; }
          W.l.o.g., $u_d$ is the operation with the highest impact (that
          is at this step $\forall j<d, a_j<_r a_d$); Reorder otherwise\;
    \For{$i=1$ to $d-1$}{
          \textit{//identify operation thresholds}\\ 
          $\mathbf{k}_{i} \gets max_{x\geq 1}(u_i^{x}(a_i)<_r
          u_d(a_d))$\;}
\textit{//$J$ is the matrix where each element $(i,j)$ is the
    impact of $k_i$ applications of $u_i$ on the $j^{th}$ centrality of node $a_i$, minus the impact of operation $u_d$ on $a_d$\;}
        Let $J_{i,j} = c^j(u_i^{k_i}(a_i)) - c^j(u_d(a_d))$ \;
        Set $J_{d,.} = 0^d$ \;
    \Return{Ker($J$)}    \textit{//find  $\mathbf{h}$ s.t. $J.\mathbf{h}=0$, thus is solution to the reverse engineering of $f$}
    \caption{A general reverse engineering algorithm, estimating the linear weight combination of centralities in $f$.}
    \label{algo}
  \end{algorithm}
\end{figure}

As the space of possible centralities is
theoretically infinite, we assume the user takes a bet on a
possibly large list of $d$ centralities in a set $C$, that are
potentially involved in $f$. We will show that our approach also allows to infer
the absence of significant impact of a given centrality in set $C$,
and thus conclude that it is probably not used in $f$.

In a nutshell, the query proceeds as follows. The user leverages an
arbitrary node $a$, already present in $G_\infty$. She then creates
$d$ identical nodes (\ie, profiles) and connects them to $a$. The
ranking of those $d$ created nodes is thus the same, by
construction. She applies to each node a different serie of API calls
(\ie, topological operations, attaching them one node for
instance). After each serie, ranking of those nodes changes. Based on
those observed changes, she is able to sort the impact of those calls,
and thus to describe the impact of one given call by a composition of
smaller effect calls. This allows her to retrieve the weights assigned
by $f$ to the $d$ centralities in set $C$, by solving a linear
equation system.

Lets consider the following image: imagine you have an old weighing
scale (that only answers ``left is heavier than right'' or vice-versa)
and a set of fruits (say berries, oranges, apples and melons) you want
to weigh. Since no ``absolute'' weighing system is available, the
solution is to weighs the fruits relatively to each other, for
instance by expressing each fruit as a fraction of the heaviest fruit,
the melon. One straightforward approach is to directly test how many
of each fruit weigh one melon.  This is the approach adopted
here. However, the problem here is that in general, we are not able to
individually weigh each fruit (centrality). Instead, we have a set of
$d$ different fruit salads. This is not a problem if the composition
of each salad is known (\ie, the impact of API calls); one has to
solve a linear system: there are $d$ different combinations that are
equal, thus providing $d$ equations.

\paragraph{A reverse engineering algorithm}
Black-box function $f$ relies on arbitrary centralities chosen from a set we denote $C$
of size $|C| = d$. Let $\mathbf{c_i}\in \mathbb{R}^d$ be the $d$ dimensional
column vector representing each of the $d$ computed centrality values
for a node $i$ in $G_{\infty}$.

We assume that $f$ is \textbf{linear} in all directions
(\ie, $f$ is a weighted mean of all centralities): $\exists
\mathbf{h}\in \mathbb{R}^d$ s.t. $f(i) =
\mathbf{c_i}.\mathbf{h}$. Reverse engineering the topological impact
over the final ranking thus boils down to find $\mathbf{h}$ (and therefore
directly obtain $f$).$\mathbf{h}$ is then the vector of coefficients corresponding to centralities listed in C.
The user performs operations on $G_{\infty}$ through API calls, starting from an existing node $a$.
We assume she is able to find $d$ different operations denoted
$\{u_1,\ldots,u_d\}$. Consider for instance one operation of that set,
noted $u_1(i)$, and that simply adds a neighbor to node $i$:
$G_{\infty}(V,E),i \rightarrow^{u_1(i)}
(V\cup\{a\},E\cup\{(i,a)\})$. Such an operation has an impact on $i$'s
topological role in $f$, let $\mathbf{u} \in \mathbb{R}^d$ be this
impact on all centralities in set $C$:
$\mathbf{c_i} \gets \mathbf{c_i}+\mathbf{u}$.

Regarding those operations, we assume that :
\textit{i)} the user is able to determine the result of each $u_i$'s impact on
  her created node's centrality values (\ie, she computes $u_i^k(i), \forall
  i\leq d$, and where $k>0$ is the number of applications of the operation), and
\textit{ii)} they are linearly independent: each operation has a
  unique impact on computed centralities from set $C$.

The query proceeds as shown on Algorithm~\ref{algo}, where notations are defined. First, observe that
by construction $rank(J)\leq d-1$. The last operation $u_d$ is the
reference against which we compare other operations.  Line 12 records the
maximum number of same $u_i$ operation applications that lead to the same rank (or
close) than a single $u_d$ operation on another node.

\begin{figure}[t!]
  \centering
$$ 
\small 
  \underbrace{ 
 \begin{pmatrix}
  j_{1,1} & j_{1,j} & \cdots & j_{1,n-1} \\
  j_{2,1} & j_{2,j} & \cdots & j_{2,n-1} \\
  \vdots  & j_{i,j}=c^j(u_i^{k_i}(a_i)) - c^j(u_d(a_d)) & \ddots & \vdots  \\
  0 & 0 & \cdots & 0 \\
 \end{pmatrix}
             }_{J}
  \underbrace{ 
 \begin{pmatrix}
 h_{0} \\
 h_{1} \\
 \vdots  \\
 h_{d-1} 
 \end{pmatrix}
             }_{\mathbf{h}}
= 0
$$

  \caption{Solution to reverse engineer $f$,  uncovering $\mathbf{h}$.}
  \label{fig:matrices}
\end{figure}

Consider a line $i$ of $J$ (L.12 Algorithm~\ref{algo}, also represented on Figure~\ref{fig:matrices}).  Since at the end $a_i =_r a_d$ (or
close), we have $(c_{a_i}+u_i^{\mathbf{k}_{i}}(a_i))\mathbf{h} =
(c_{a_d}+u_d(a_d))\mathbf{h} \pm u_i(a_i).\mathbf{h}$. Since by construction
$c_{a_i}=c_{a_d}$, therefore we seek $\mathbf{h}$
s.t. $u_i^{\mathbf{k}_{i}}(a_i)\mathbf{h} - u_d(a_d)\mathbf{h}= 0$. Or
 matrix notation: $J.\mathbf{h} =0$: $\mathbf{h}$ is in the kernel
of $J$.

Intuitively, the fact that we get infinitely many solutions
($\alpha.\mathbf{h}, \forall \alpha \in \mathbb{R}^+$) comes from our
observation method: we are never able to observe actual scores,
but rather rankings. Since multiplying $\mathbf{h}$ by a constant does not change the final ranking, any
vector co-linear to, \eg, $\mathbf{h}/||\mathbf{h}||$ is a solution.

One important remark is that one cannot formally claim that one
centrality metric is not in use in $C$ with this algorithm. Assume for
instance that one centrality, \eg, number of tweets of the considered
node, is $10^3$ times less important than another centrality, \eg,
degree. Then we will not be able to witness its effect unless we
produce $10^3$ tweets. And after $10^2$ tweets, we will only be able
to conclude: number of tweets is at least $10^2$ times less important
that degree.  One can reasonably assume that such an imbalance in
practice means that one service operator will not compute a possibly
costly centrality to use it to such a low extent in $f$; this thus
makes our algorithm able to discard barely or not used centralities in
$C$.
Finally, we note that with $cost(u_i)$ the number of calls issued by
operation $u_i$, the total number of operations for weighting
two centralities in $C$ is at most $cost(u_d)+ \sum_{i=1}^{d-1}
k_i.cost(u_i)$.

\paragraph{Exploiting local centralities: an illustration}

We demonstrated how to reverse engineer a linear
combination of centralities. The difficulty for the user is to compute
the impact of $u$ operations on the suspected centralities. In the
easiest case, suspected centralities behave linearly (such as \eg,
degree, betweenness), and it is therefore easy to compute the impact
of an update. The case of non-linearity can be solved using the
locality of centralities: if $c$ is $k$-local, the observation of the
$k$-hop neighborhood of a node is required to reverse engineer $f$. We
illustrate this on a simple example.

\begin{figure}[h!]
  \centering
         \begin{tikzpicture}[shorten >=1pt,auto]
           \tikzstyle{vertex}=[circle,fill=black!25,minimum
size=17pt,inner sep=0pt]
           \tikzstyle{graph}=[circle,fill=black!10,minimum size=23pt,inner
           sep=0pt]
           \tikzstyle{upda}=[star,line width= 5pt,
color=blue,fill=black!10!white,minimum size=23pt,inner
           sep=0pt,star points=7]
           \node[graph] (g) at (0,0.8) {$G_{\infty}$};
           \node[vertex] (a) at (0,0) {$a$};
           \node[vertex] (n1) at (-2,-1) {$a_1$};
           \node[vertex] (n2) at (2,-1) {$a_2$};
           \foreach \from/\to in {a/g,a/n1,a/n2}
           \draw (\from) -- (\to);

\node[vertex] (n1u1) at (-2.8,-2) {$a_{11}$}; 
\node[vertex] (n1u2) at (-1.8,-2) {$a_{12}$}; 
\node[vertex] (n1u3) at (-0.8,-2) {$a_{13}$};
             \draw (n1) -- (n1u1);
             \draw (n1) -- (n1u2);
             \draw (n1) -- (n1u3);
             \draw[]  (n1u1) to node [auto] {} (n1u2);
             \draw[in=30,]  (n1u1) to node [auto] {} (n1u3);

\node[vertex] (n2u1) at (0.7,-2) {$a_{21}$}; 
\node[vertex] (n2u2) at (1.7,-2) {$a_{22}$}; 
\node[vertex] (n2u3) at (2.7,-2) {$a_{23}$};
             \draw (n2) -- (n2u1);
             \draw (n2) -- (n2u2);
             \draw (n2) -- (n2u3);

             \node[upda] (u1) at (-1.5,-3) {$u_1$};
             \node[upda] (u2) at (2,-3) {$u_2$};
         \end{tikzpicture}
  \caption{Querying $G_{\infty}$: conducing two sequences of operations $u_1$ and $u_2$, attaching them to $a$.}
  \label{fig:general}
\end{figure}
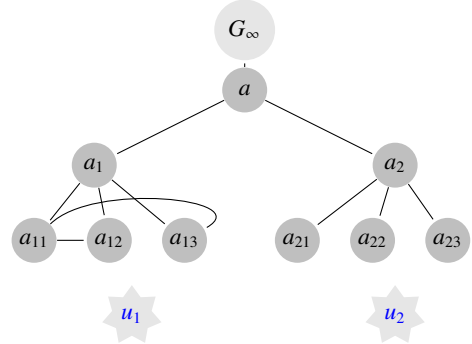

Let us assume a ranking function $f$ whose internals use a combination
of $c_1$: $clustering$ centrality\footnote{this centrality has no
    linear behavior, but is $1-$local.} and $c_2$: $degree$ (\ie,
$C=[c_1,c_2]$). Without loss of generality, we assume that the
coefficient for degree in $\mathbf{h}$ is $h_1=1$, so that we seek the
corresponding coefficient $h_2=h$.  Let us consider the following two
operations in $\{u_1,u_2\}$.
Operation $u_2$ simply attach a node to one initial query node
  ($a1$ or $a2$). $u_1$ starts by attaching $querySize-1$ nodes to
  query node $a1$. At each call, an edge between two
  randomly selected attached nodes is added, to increase clustering. $u_1$ and
  $u_2$ are represented on Figure~\ref{fig:general}, for a
  $querySize=4$. User can compute the value of
  $u_1^{k-1}(a_1)$ and $u_2^{k-1}(a_2)$ at any time, since she
  controls those nodes.  

We simulated the query with a $G_\infty$ being a $1,000$ nodes
Barab\'asi-Albert graph with an average degree of $5$, estimating $h$
using $u_1$ and $u_2$ operations with Algorithm 1. Figure
\ref{fig:sub2} presents the obtained results: a point $(x,y)$ means
the real value of $h$ is $x$ and was estimated by Algorithm 1 as
$y$. Black dots plot the real coefficient values of $h$. Each colored
area represents the estimated (reverse engineered) coefficients, while
each color represents a query size, \ie, the number of nodes created
by the user to reverse $f$. The larger the query, the more
precise the reverse engineered results. We note that if the real
values of coefficients to be estimated are bigger (e.g., 4 or 5 on the
$x$-axis), estimations show lower precision (larger areas). Despite
this remark, estimations appear unbiased.

\section{Discussion}
The will for web-services transparency starts to trigger new research
works. Security-oriented paper~\cite{stealing} has shown that it is
possible to ``steal'' some machine learning models from online
services, using a reasonable number of queries to APIs.
XRay~\cite{chaintreau} for instance proposes a correlation algorithm,
that aims at inferring to which data input is associated a
personalized output to the user. This Bayesian-based algorithm returns
data that are the cause of received ads, while we seek in this paper
to retrieve the internals of a black-box ranking function, in order to
assess what is the effect of user actions on the output peer ranking.
We have presented a general framework. Based on the centralities that
might be used by the ranking function, there remain work for a user,
for building discriminating graphs, and for finding small topological
operations that will make the reverse engineering possible.  
For a
ranking service operator, the countermeasure is the opposite: she must
find ranking metrics that are computationally hard to distinguish,
typically ones that would ensure the detection of the querying user by
the internal security system. We find this to be an interesting
challenge for futureworks
\begin{figure}[t!]
  \centering
  \includegraphics[width=1\linewidth]{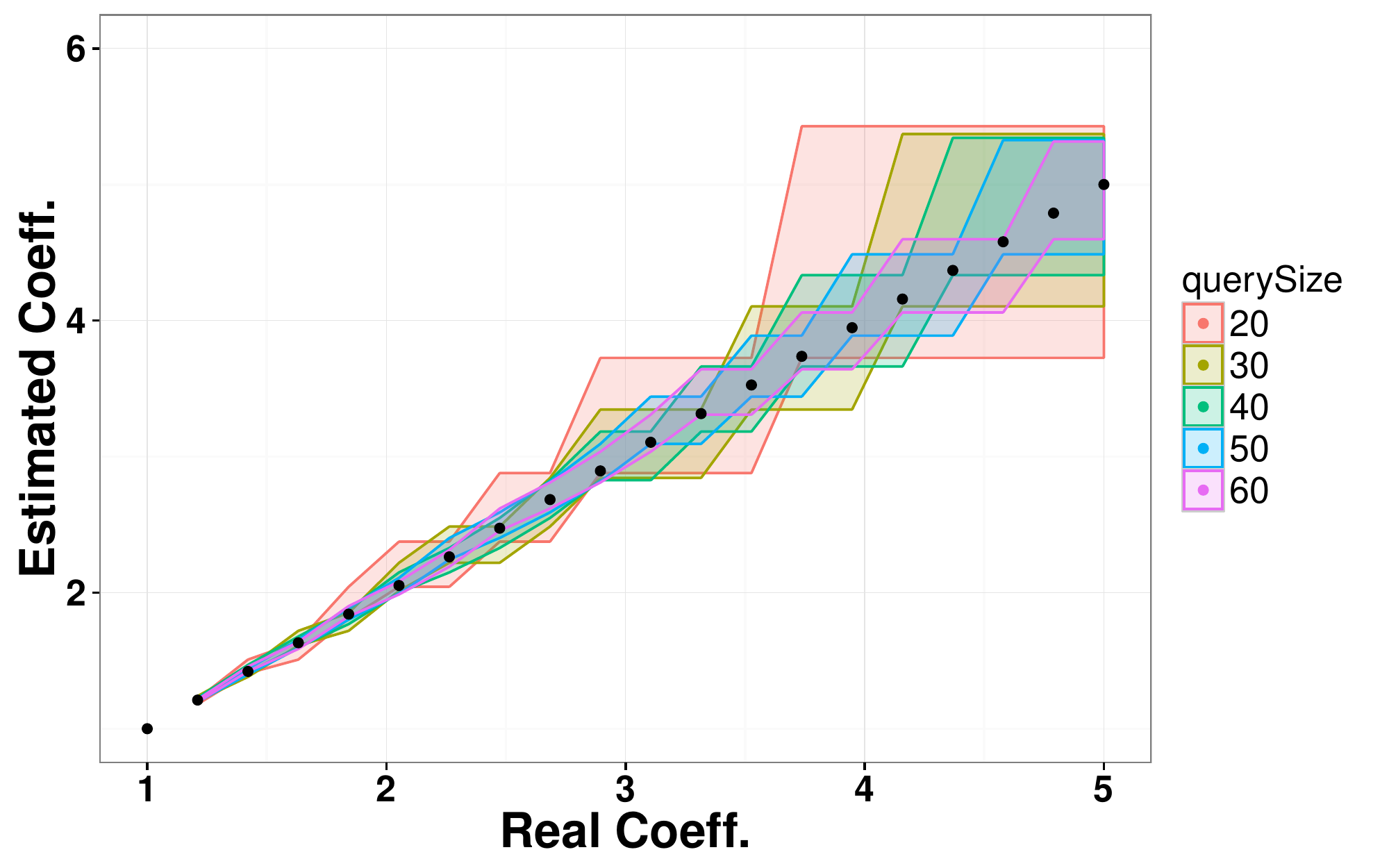}
  \caption{Reversing a $f$ with unknown coefficients from $1$ to $5$, with various query sizes (node creations.)}
  \label{fig:sub2}
\end{figure}

\bibliographystyle{abbrv}
\bibliography{biblio}

\end{document}